\documentclass{style/vldb_style_sample/vldb}

\usepackage{graphicx}
\usepackage{balance}  

\usepackage{booktabs}
\usepackage{amsmath}
\usepackage{verbatim}
\usepackage{algorithmicx}
\usepackage[plain]{algorithm}
\usepackage{algpseudocode}
\usepackage{pgfplots}
\usepackage{pgfplotstable}
\usepackage[center]{subfigure}
\usepackage{url}

\usetikzlibrary{patterns,shapes,arrows,positioning,fit,decorations.pathreplacing}

\usepgfplotslibrary{groupplots}

\newlength{\abovecaptionskip}
\newfont{\mycrnotice}{ptmr8t at 7pt}
\newfont{\myconfname}{ptmri8t at 7pt}
\let\crnotice\mycrnotice%
\let\confname\myconfname%

\begin{document}

\newtheorem{lemma}{Lemma}
\newtheorem{theorem}{Theorem}
\newtheorem{corollary}{Corollary}
\newtheorem{example}{Example}
\newtheorem{assumption}{Assumption}
\newtheorem{remark}{Remark}

\newdef{definition}{Definition}
\newdef{scenario}{Scenario}


\title{Parallelizing Query Optimization \\on Shared-Nothing Architectures\thanks{This work was supported by ERC Grant 279804 and by a European Google PhD fellowship.}}



%
%
%
%

\numberofauthors{1} 

\author{
%
%
\alignauthor
Immanuel Trummer and Christoph Koch\\
			 \email{\{firstname\}.\{lastname\}@epfl.ch}\\
       \affaddr{\'Ecole Polytechnique F\'ed\'erale de Lausanne}
}

\maketitle

\newcommand*{\codeF}{\fontfamily{\sfdefault}\selectfont}
\newcommand*{\branchF}{\fontsize{9}{9}\selectfont}
\newcommand*{\connectorF}{\fontsize{9}{9}\selectfont}

\begin{abstract}
Data processing systems offer an ever increasing degree of parallelism on the levels of cores, CPUs, and processing nodes. Query optimization must exploit high degrees of parallelism in order not to gradually become the bottleneck of query evaluation. We show how to parallelize query optimization at a massive scale.

We present algorithms for parallel query optimization in left-deep and bushy plan spaces. At optimization start, we divide the plan space for a given query into partitions of equal size that are explored in parallel by worker nodes. At the end of optimization, each worker returns the optimal plan in its partition to the master which determines the globally optimal plan from the partition-optimal plans. No synchronization or data exchange is required during the actual optimization phase. The amount of data sent over the network, at the start and at the end of optimization, as well as the complexity of serial steps within our algorithms increase only linearly in the number of workers and in the query size. The time and space complexity of optimization within one partition decreases uniformly in the number of workers. We parallelize single- and multi-objective query optimization over a cluster with 100 nodes in our experiments, using more than 250 concurrent worker threads (Spark executors). Despite high network latency and task assignment overheads, parallelization yields speedups of up to one order of magnitude for large queries whose optimization takes minutes on a single node. 
\end{abstract}


\section{Introduction}
\label{introSec}

Moore's law~\cite{Moore1998} is breaking and computer systems become more powerful by increasing their number of processing units (be it cores, CPUs, or cluster nodes) rather than by increasing clock rates. This means that all stages of query evaluation must exploit parallelism in order not to become the bottleneck in future systems. 

Research on parallelizing query evaluation has so far mainly focused on how to parallelize the actual query processing stage, i.e.\ how to parallelize the execution of query plans. This is however insufficient as noted in prior work~\cite{Han2008, Han2009, Waas2009, Soliman2014}: in order to parallelize query evaluation, we must not only parallelize the execution of query plans but also the \textit{generation} of query plans, i.e.\ we must develop parallel algorithms for the query optimization problem. 

Query optimization is an NP-hard problem and even finding guaranteed near-optimal query plans is NP-hard~\cite{Chatterji2002}. The run time of all known algorithms increases exponentially in the number of joins and novel application scenarios (e.g., SPARQL query processing~\cite{Cure2014}) motivate queries with many joins. Furthermore, the complexity of the systems on which query processing takes place increases: the number of system components keeps increasing (as discussed before), flexible provisioning models and novel processing operators introduce new parameters by which query processing can be tuned (e.g., the number of machines to rent is such a parameter in a cloud scenario~\cite{kllapi2011schedule} or the sampling rate of a scan operator in the context of approximate query processing~\cite{Agarwal2012}). All those developments make query optimization harder since the size of the plan search space increases. In addition, many of the aforementioned developments motivate new cost metrics for comparing query plans (e.g., monetary fees in a cloud scenario or result precision in approximate query processing) in addition to execution time. Having multiple plan cost metrics makes query optimization however harder as well~\cite{Trummer2014, Trummer2015a, Trummer2015}. In summary, there are many ongoing developments that make query optimization harder and hence increase the need for parallel query optimization algorithms.

We propose a novel, parallel algorithm for query optimization in this work. Our goal is to obtain a query optimization algorithm that is future-proof in that it is able to exploit the ever-growing degree of parallelism forced by the breakdown of Moore's law. While prior parallel query optimization algorithms have been primarily designed for shared-memory architectures, we aim at parallelizing query optimization on shared-nothing architectures as well. Query plans are often executed on large clusters and, as query optimization must  precede query execution, it is preferable to use all cluster nodes for query optimization rather than leave them idle until optimization has finished. Even for queries that are executed repeatedly on a single node, a cluster can be used for optimization before run time if optimization is expensive. The algorithm that we propose is however not specific to shared-nothing architectures and can be applied in different scenarios as well.

Prior approaches for parallelizing query optimization assume that worker threads share common data structures~\cite{Han2008, Han2009, Waas2009, Chen2012a, Soliman2014}, in particular big memotables storing subsets of query tables optimal join plans. They assume that a central master node distributes fine-grained optimization tasks to workers and that many interactions between master and worker threads take place during the optimization of a single query. In a shared-nothing architecture, sharing data between worker threads results in high communication overhead and each task assignment incurs setup overhead. We target extremely high degrees of parallelism, at least several hundreds of cluster nodes (while prior algorithms have not been evaluated on more than eight cores). Orchestrating that many nodes on the level of micro optimization tasks results in prohibitive communication and computation overhead on the master node. 

Achieving our goals requires a radically different approach compared to prior work: instead of decomposing the query optimization problem into many small optimization tasks, we realize the most coarse-grained problem decomposition possible: the optimization of one query is mapped into exactly one task per worker node.

On a high level, our algorithm works as follows. Given a query to find an optimal plan for, the master optimizer node sends that query together with a plan space partition ID to each worker node. The partition ID is simply an integer between one and the number of workers such that each worker obtains a different number. Each worker node translates its partition ID into a set of constraints on join orders and only considers query plans that comply with those constraints. Each worker node therefore searches for an optimal plan within a plan space that is smaller than the original plan space. The worker nodes search the optimal plan within their respective plan space partition in parallel. No communication between workers or between workers and master node is required during that stage. Afterwards, the workers send the optimal plans back to the master node. The original plan space is the union over all plan space partitions. Comparing the plans returned by the workers, which are optimal within their respective partition and whose number is linear in the number of workers, yields therefore the globally optimal plan.

Our algorithm is designed to exploit very high degrees of parallelism. The time complexity of all serial processing steps, executed by the master node, is linear in the number of workers and in the query size. The amount of data sent over the network is also linear in the number of workers and in the query size. All plan space partitions have the same size which guarantees skew-free parallelization. For a fixed query, the run time as well as the consumption of main memory space per worker node decreases monotonically in the number of worker nodes. Furthermore, the number of partitions into which the plan space can be divided and therefore the maximal degree of parallelism grows in the query size and is in principle unlimited. 

Our algorithm parallelizes one of the most popular dynamic programming schemes for query optimization~\cite{Selinger1979}. It treats table sets of increasing cardinality and constructs optimal join plans for each table set out of optimal plans for table subsets that were previously generated. As it has been noted in prior work~\cite{Han2008}, this dynamic programming scheme belongs to the class of non-serial polyadic algorithms and is therefore difficult to parallelize. Certainly it is easier to parallelize randomized query optimization algorithms such as iterated improvement or simulated annealing~\cite{Swami1989, Ioannidis}. We nevertheless focus on parallelizing the dynamic programming approach. There are two reasons. First, unlike randomized algorithms, the dynamic programming approach formally guarantees to return optimal query plans. Second, by parallelizing Sellinger's classical dynamic programming scheme~\cite{Selinger1979} we parallelize at the same time many query optimization algorithms that have been based on the same scheme and cover a multitude of scenarios (e.g., multi-objective query optimization~\cite{Trummer2014, Trummer2015a} or parametric query optimization~\cite{Hulgeri2003}).

The time and space complexity of the classical dynamic programming algorithm depend on the number of table sets for which optimal join plans need to be found. We decompose the query optimization problem by introducing constraints on the join order that ultimately allow to reduce the number of table sets to consider.

We propose a partitioning scheme for the space of left-deep query plans and one partitioning method for bushy query plans. Left-deep query plans are characterized by the order in which tables are joined. We restrict join orders by constraints of the form $x\prec y$ where $x$ and $y$ are query tables: the semantics is that table $x$ needs to be joined before table $y$. The constraint excludes any query plan producing an intermediate join result containing table $y$ but not table $x$ and hence we can neglect table sets containing $y$ without $x$ during dynamic programming. This reduces the number of table sets to consider by a factor of 3/4. If we assign the constraint $x\prec y$ to a first worker node and the complementary constraint $y\prec x$ to a second worker then the entire search space is covered. Furthermore, we can recursively decompose the resulting plan space partitions by applying similar constraints to other (disjoint) table pairs. 




Bushy query plans are binary trees and cannot be represented as join orders anymore. However, if we fix an arbitrary table and follow its way from a leaf node in the plan tree to the root then we can order the other tables based on when they first appear in the sequence of intermediate results we encounter. Hence we restrict join orders for bushy plan spaces by constraints of the form $x\preceq y|z$ with the semantics that $x$ appears no later than $y$ when following table $z$ to the plan tree root. This excludes join results that contain tables $y$ and $z$ but not table $x$. 

We formally analyze time and space complexity and the network bandwidth required by our algorithm. We show that each constraint reduces time complexity by factor 3/4 for linear and by factor 21/27 for bushy plan spaces. We show that those reduction factors are actually optimal within a restricted design space of partitioning methods. Prior algorithms achieved near linear speedups until a low number of threads within a shared-memory architecture. Our speedups are not linear but very steady up to very high degrees of parallelism and within a shared-nothing architecture. In our experiments, we demonstrate continuous scaling up to more than 250 concurrent worker threads on a large cluster over various query sizes and for single as well as multi-objective query optimization. As our algorithm scales even in this challenging scenario, we believe that it scales on many other architectures as well.

The original scientific contributions of this paper are in summary the following:

\begin{itemize}
\item We propose a novel algorithm for massively-parallel query optimization on shared-nothing architectures. 
\item We formally evaluate that algorithm in terms of time and space complexity and in terms of the required network traffic.
\item We evaluate the algorithm experimentally on a large cluster, demonstrating its scalability for up to more than 250 concurrent worker threads.
\end{itemize}

The remainder of this paper is organized as follows. We compare against related work in Section~\ref{relatedSec}. In Section~\ref{modelSec}, we introduce our formal problem model. We present our algorithms for parallel query optimization in left-deep and bushy plan spaces in Section~\ref{algSec}. In Section~\ref{analysisSec}, we analyze time and space complexity as well as the growth in network traffic. In Section~\ref{experimentsSec}, we experimentally demonstrate the scalability of our algorithms on a large cluster.


\section{Related Work}
\label{relatedSec}

The term \textit{parallel query optimization} sometimes refers to serial optimization algorithms generating plans that are executed in parallel~\cite{Chekuri1995}. It is crucial to realize that we use the term in a very different sense: we propose a parallel algorithm for generating query plans (that may be executed serially or in parallel). 

Our work connects to prior work that parallelizes the classical dynamic programming based query optimization algorithm~\cite{Han2008, Han2009, Waas2009, Zuo2011, Chen2012a, Soliman2014}. Prior algorithms have however implicitly been designed for shared-memory architectures that do not scale beyond a certain degree of parallelism~\cite{Stonebraker1986}. Prior algorithms have been evaluated on up to maximally eight cores while we demonstrate scalability of our algorithm on a shared-nothing architecture using over 250 workers. We outline some of the factors that distinguish prior algorithm from our algorithm and limit their scalability.

Prior algorithms assume that all threads share common data structures (e.g., the memotable containing optimal partial plans) and hence can access intermediate results generated by other threads. This would lead to huge communication overhead on shared-nothing architectures (e.g., the size of the memotable is exponential in the query size) while our algorithm does not require any communication between workers. Furthermore, prior algorithms use a central coordinator which assigns rather fine-grained optimization tasks to worker threads (e.g., the master thread assigns specific pairs of join operands to generate plans for). This has two disadvantages. First, a lot of communication is required between master and workers. Second, the time complexity for managing the workers is high, so the master itself will eventually become the bottleneck as the degree of parallelism increases. 

We assign tasks at the coarsest possible level: each worker receives exactly one task per query. The time complexity of the algorithm executed on the master is linear in the number of worker nodes and in the query size and so is the total amount of data that needs to be sent over the network. Finally, only one round of communication between workers and master is required per query by our algorithm while prior algorithms usually require many rounds of communication. Having only one round of communication is advantageous in scenarios where distributing tasks to workers and receiving the results is associated with overheads. We compare against a typical representative of prior algorithms in our experimental evaluation.


Our work is generally relevant for all areas of query optimization in which algorithms based on dynamic programming have been proposed. This includes, for instance, multi-objective query optimization~\cite{Trummer2014, Trummer2015}, parametric query optimization~\cite{Ganguly1998, Ioannidis1997}, and multi-objective parametric query optimization~\cite{Trummer2015}. Our method of partitioning the join order space is generic and can be applied to all of those scenarios. 

\section{Problem Model}
\label{modelSec}

As it is standard in query optimization, we use a simplified query and query plan model to describe our algorithms. Extending the model and the algorithms towards richer query languages and plan spaces is however straightforward and can be achieved via standard techniques~\cite{Selinger1979}. 

A query is a set $Q$ of tables that need to be joined. We denote by \textproc{Scan}($q$) for $q\in Q$ a query plan that scans a single table and call such a plan a \textit{scan plan}. By \textproc{Join}($p_L,p_R$) we designate a plan that joins the result produced by plan $p_L$ with the result produced by $p_R$ and uses $p_L$ as outer and $p_R$ as inner operand. We use the terms left and right operand as synonyms for outer and inner operand respectively as the outer operand is usually drawn at the left side in visual representations of query plans. Note that we do not incorporate alternative operator implementations for scans and joins into our model to simplify the presented pseudo-code. The extension is however easy and the implementation of our algorithm used for the experiments considers all standard operators. 

We distinguish two types of query plans. \textit{Left-deep plans} are plans in which the right operand of every join is a scan plan. All other plans are \textit{bushy plans}. Bushy plans can be represented as labeled binary trees where leaf nodes correspond to single tables and inner nodes correspond to join results. The tree shape of left-deep plans is fixed and the join order of a left-deep plan is fully described by the order in which table leaf nodes are encountered in a traversal (e.g., in post-order) of the plan tree. This is why we can represent left-deep plans by a sequence of query tables.

For a fixed query, the set of all bushy plans is the \textit{bushy plan space} and the set of all left-deep plans is the \textit{left-deep} or \textit{linear plan space}. We assume that a cost model is available that associates query plans with cost estimates. Our pseudo-code encapsulates that cost model in a pruning function that discards the plan with higher cost among several compared plans. The goal of query optimization is to find the cost-optimal plan either in the space of left-deep or in the space of bushy plans. 
\section{Algorithm}
\label{algSec}

We present an algorithm for massively-parallel query optimization. The algorithm is well suited for shared-nothing architectures as it minimizes the amount of sychronization and communication overhead. The same properties are however beneficial in shared-memory scenarios. Our algorithm is not specific to shared-nothing architectures and can be used to parallelize query optimization over the nodes of a cluster or over the cores of a single computer all the same. 

The presented algorithm solves the traditional query optimization problem, meaning that it compares alternative query plans according to single point cost estimates in one cost metric. The method by which we partition the plan space is however very generic and it is in fact straight-forward to extend our algorithm to handle multiple plan cost metrics~\cite{Trummer2014, Trummer2015a} or plan cost functions that depend on unknown parameters~\cite{Ioannidis1997, Ganguly1998} or both together~\cite{Trummer2015}. This is possible since algorithms have been proposed for all of the aforementioned query optimization variants that use the same dynamic programming scheme as the classical algorithm by Selinger~\cite{Selinger1979}; only the pruning function, the way in which different query plans are compared, differs between them. The algorithm presented next can therefore easily be transformed into an algorithm handling other query optimization variants by essentially replacing the pruning function.

We present two variants of our algorithm: the first variant finds the optimal left-deep query plan for a given query while the second variant finds the optimal plan within a bushy plan space. Before discussing the pseudo-code, we illustrate informally how our algorithm works by means of a simplified example. This example refers to the algorithm variant searching left-deep plan spaces.

\begin{example}
Assume we want to find the optimal left-deep plan for answering the join query $R\Join S\Join T\Join U$. Further assume that four worker nodes are available over which query optimization is parallelized. Upon reception of the query, the master nodes sends the query together with the total number of plan space partitions (four) and the respective partition ID (between one and four) to each worker node. Consider the worker node that partition three is assigned to. Knowing that the total number of partitions is four, the worker node derives that it should use $\log_2 4=2$ constraints to restrict the join order space. The two constraints refer to the order in which the four tables are joined. The first constraint refers to the ordering between the first pair of tables, $R$ and $S$, and establishes which of them appears first in the join order. The second constraint refers to $T$ and $U$. The binary representation of the partition ID encodes the concrete set of constraints to use. For the considered worker node, the partition ID is $10$ in binary representation. The first bit of the binary representation is zero so the worker node orders $R$ before $S$. As the second bit is one, the worker orders $U$ before $T$. Note that other workers will use complementary constraint sets based on their respective partition ID such that the whole join order space is covered. The worker that we focus on finds the best plan whose join order complies with the given constraints. It returns that plan to the master which compares the plans returned by all workers to determine the globally optimal plan.
\end{example}

We present pseudo-code for the high-level algorithm that is executed by the master and the worker nodes in Section~\ref{highLevelSub}. The code of the sub-functions that the workers use to infer constraints on the join order from the partition ID and to find join orders that comply with the constraints are discussed in Section~\ref{partitioningSub}.

\subsection{High-Level Algorithm}
\label{highLevelSub}

We present pseudo-code for the high-level algorithms that are executed on the master node and on the workers. As it is common in the area of query optimization, we simplify the presented pseudo-code by considering only SPJ queries and by neglecting for instance the impact of interesting tuple orders~\cite{Selinger1979}. There are however standard methods by which such algorithms can be extended to support richer query languages~\cite{Selinger1979} (e.g., queries with aggregates or nested queries). It is straight-forward to extend the presented algorithm to consider interesting tuple orderings, too.

As announced before, we present two algorithm variants, one treating the space of left-deep plans, the other one treating the space of bushy plans. The pseudo-code that we discuss in this subsection is however the same for both variants such that we do not need to distinguish between them.

\algblock{ParFor}{EndParFor}
\algnewcommand\algorithmicparfor{\textbf{parfor}}
\algnewcommand\algorithmicpardo{\textbf{do}}
\algnewcommand\algorithmicendparfor{\textbf{end\ parfor}}
\algrenewtext{ParFor}[1]{\algorithmicparfor\ #1\ \algorithmicpardo}
\algrenewtext{EndParFor}{\algorithmicendparfor}

\begin{algorithm}[t!]
\renewcommand{\algorithmiccomment}[1]{// #1}
\begin{algorithmic}[1]
\State \Comment{Parallelizes optimization of query $Q$ over $m$ machines.}
\Function{Master}{$Q,m$}
\State \Comment{Generate best plan for each partition in parallel}
\ParFor{$partID\in \{1,\ldots,m\}$}
\State $bestInPart[partID]\gets$\Call{Worker}{$Q,partID,m$}
\EndParFor
\State \Comment{Prune plans and returns best plan}
\State $bestPlan\gets bestInPart[1]$
\For{$partID\in \{2,\ldots,m\}$}
\State \Call{FinalPrune}{$bestPlan,bestInPart[partID]$}
\EndFor
\State \Return{$bestPlan$}
\EndFunction
\end{algorithmic}
\caption{Function executed by master node for parallel query optimization on shared-nothing architectures.\label{masterAlg}}
\end{algorithm}

Our algorithm consists of two parts: the first part is executed by the master node which orchestrates the worker nodes. The second part of our algorithm runs on the worker nodes. Algorithm~\ref{masterAlg} shows the code that is executed on the master. The input is a query $Q$, for which we want to find an optimal query plan, and the number $m$ of available worker nodes. We assume in the following that $m$ is a power of two (the reason will become apparent in the following). The output of the \textproc{Master} function is the optimal plan for $Q$. 

The master node executes two phases. In the first phase, the master sends the query together with a unique partition ID to each of the workers. We discuss the pseudo-code of the \textproc{Worker} function a bit later. All worker invocations happen in parallel as indicated by the keyword \textbf{parfor}. The partition ID identifies a partition of the plan search space. The task of each worker is to find the optimal plan within its respective partition and to return it to the master. The master collects the returned plans in the array $bestInPart$ (we use the standard notation $bestInPart[x]$ to represent an access to the $x$-th field of that array). In the second phase, the master node compares all collected plans to identify the globally-optimal plan. Function~\textproc{FinalPrune}, whose pseudo-code we do not specify, represent a standard pruning function that replaces $bestPlan$ by the better plan among the two input plans. Having considered all plans returned by the workers, the best plan must be globally optimal.

\begin{algorithm}[t!]
\renewcommand{\algorithmiccomment}[1]{// #1}
\begin{algorithmic}[1]
\State \Comment{Generate best plan for query $Q$ in partition with }
\State \Comment{ID $partID$ out of $m$ partitions.}
\Function{Worker}{$Q,partID,m$}
\State \Comment{Decode partition ID into a set of constraints}
\State $constr\gets$\Call{PartConstraints}{$Q,partID,m$}
\State \Comment{Generate admissible intermediate results}
\State $joinRes\gets$\Call{AdmJoinResults}{$Q,constr$}
\State \Comment{Initialize best plans for single tables}
\For{$q\in Q$}
\State $P[q]\gets$\Call{Scan}{$q$}
\EndFor
\State \Comment{Iterate over join result cardinality}
\For{$k\in\{2,\ldots,|Q|\}$}
\State \Comment{Iterate over admissible join results}
\For{$q\in joinRes:|q|=k$}
\State \Comment{Try splits of $q$ into two join operands}
\State \Call{TrySplits}{$q,constr,P$}
\EndFor
\EndFor
\State \Comment{Return best plan for query $Q$}
\State \Return{$P[Q]$}
\EndFunction
\end{algorithmic}
\caption{Generate best query plan within specific partition of either linear or bushy plan space.\label{WorkerAlg}}
\end{algorithm}

Algorithm~\ref{WorkerAlg} shows the code of the function that runs on worker nodes and is invoked by the master. The input is the query $Q$ to optimize, the total number $m$ of plan space partitions, and the identifier $partID$ of the partition that is assigned to the respective worker. The output is the optimal plan within the corresponding partition. Each worker node executes the following three steps. First, knowing the total number $m$ of partitions, the specific partition ID $partID$ can be translated into a set of constraints on the join order. Function~\textproc{PartConstraints}, whose code is discussed later, accomplishes the translation. Second, function~\textproc{AdmJoinResults} translates the set of constraints into an admissible set of table sets that can appear as join results within a query plan whose join order respects the constraints. Finally, the worker node uses a dynamic programming approach to find the optimal query plan among all plans that produce only admissible join results. We assume, without explicitly writing out the corresponding code, that the result sets generated by function~\textproc{AdmJoinResults} have been indexed by their cardinality such that Algorithm~\ref{WorkerAlg} can efficiently retrieve all sets with a given cardinality $k$.


Variable $P$ is an array storing optimal query plans and $P[Q]$ designates the optimal query plan for joining the table set $Q$. We initialize $P$ by inserting the scan plan for each single query table $q\in Q$. We simplify the pseudo-code by assuming only one scan plan per table but the generalization is straight-forward. After that, the algorithm calculates optimal plans for table sets of increasing cardinality, using the optimal plans that were stored in prior iterations. The algorithm considers only table sets that represent admissible join results. For each admissible join result, function~\textproc{TrySplits} tries all ways of splitting the join result into two admissible operands and stores the best resulting plan in $P$. 

\subsection{Plan Space Partitioning}
\label{partitioningSub}

\begin{algorithm}[t!]
\renewcommand{\algorithmiccomment}[1]{// #1}
\begin{algorithmic}[1]
\State \Comment{Generate constraint on $i$-th table pair of }
\State \Comment{query $Q$ using precedence order $precOrd$.}
\Function{Constraint[Linear]}{$Q,i,precOrd$}
\If{$precOrd=0$}
\State \Return{$Q_{2\cdot i}\prec Q_{2\cdot i+1}$}
\Else
\State \Return{$Q_{2\cdot i+1}\prec Q_{2\cdot i}$}
\EndIf
\EndFunction
\vspace{0.1cm}
\State \Comment{Generate constraint on $i$-th table tuple of}
\State \Comment{query $Q$ using precedence order $precOrd$.}
\Function{Constraint[Bushy]}{$Q,i,precOrd$}
\If{$precOrd=0$}
\State \Return{$Q_{3\cdot i}\preceq Q_{3\cdot i+1}|Q_{3\cdot i+2}$}
\Else
\State \Return{$Q_{3\cdot i+1}\preceq Q_{3\cdot i}|Q_{3\cdot i+2}$}
\EndIf
\EndFunction
\vspace{0.1cm}
\State \Comment{Decode partition ID $partID$ into a set of constraints}
\State \Comment{restricting the plan space for query $Q$. The total}
\State \Comment{number of partitions is $m$ and $partID\leq m$.}
\Function{PartConstraints}{$Q,partID,m$}
\State \Comment{Initialize constraint set}
\State $constr\gets\emptyset$
\State \Comment{Iterate over constraints}
\For{$i\in\{0,\ldots,\log_2(m)-1\}$}
\State \Comment{$i$-th bit encodes precedence order}
\State $precOrd\gets$\Call{Bit}{$partID,i$}
\State \Comment{Generate constraint on $i$-th subset of $Q$}
\State $c\gets$\Call{Constraint}{$Q,i,precOrd$}
\State \Comment{Add new constraint into set}
\State $constr\gets constr\cup c$
\EndFor
\State \Return{$constr$}
\EndFunction
\end{algorithmic}
\caption{Translate the partition ID into a set of constraints that restrict the plan search space.\label{translateAlg}}
\end{algorithm}

We discuss in the following the sub-functions that are invoked by the \textproc{Worker} function. In contrast to the previous subsection, we now need to distinguish between the two algorithm variants that we present. In the following pseudo-code, we use the notation \textproc{F}$[LINEAR]$ to indicate that function \textproc{F} is specific to the algorithm searching linear (or left-deep) search spaces. Analogously, \textproc{F}$[BUSHY]$ indicates a function that is specific to the algorithm generating bushy plans. The code of all other functions is the same for both variants. 

Algorithm~\ref{translateAlg} shows the code for translating a partition ID into a set of constraints. Function~\textproc{PartConstraints} obtains as input the query, the number of partitions, and the partition ID. The output is a set of constraints on the join order that define the plan space partition that the current worker needs to treat. 

When generating constraints, we use the notation $Q_x$ with $x\in\mathbb{N}$ to designate the $x$-th table in query $Q$. This notation assumes that query tables have been numbered consecutively from 0 to $|Q|-1$. The algorithm can use an arbitrary table numbering but it is important that all workers use the same numbering in order to guarantee that the whole plan space is covered by the ensemble of workers. 

The form of the generated constraints differs depending on whether we search for left-deep or bushy plans. Constraints for the left-deep plan space are defined on table pairs while constraints on bushy plans are defined on triples of tables. Constraint restricting the linear plan space are of the form $Q_x\prec Q_y$. This means that the $x$-th table must appear before the $y$-th table in an admissible join order (the join order of a left-deep plan can be represented as a sequence of tables and the constraints refer to that representation). Constraints restricting bushy plan spaces are of the form $Q_x\preceq Q_y|Q_z$ with the semantic that when considering the intermediate join results containing table $Q_z$ in ascending order of cardinality, table $Q_y$ must not appear before table $Q_x$. We assume that constraints have been indexed such that all constraints concerning a given set of tables can be retrieved efficiently.

In case of a left-deep plan space there are two complementary constraints for each pair of tables, namely $Q_x\prec Q_y$ and $Q_y\prec Q_x$. In order to guarantee that the whole plan space is covered by the ensemble of workers, we need to consider complementary constraints by different workers. All workers use constraints on the same table pairs but the direction of those constraints (which of the two tables to join first) differs among workers. Each worker uses the binary representation of the partition ID to derive which of the two possible constraints to consider for each table pair. We use the notation \textproc{Bit}($partID,i$) to represent the $i$-th bit of the binary representation (it does not matter whether we start with the lowest order bits or with the highest order bits). Each bit determines the direction for one constraint. 

The treatment of bushy plan spaces is analogue. Constraints are defined on table triples but for each triple of tables there are still just two complementary constraints and each worker picks between them based on the partition ID. We define two variants of the function~\textproc{Constraint} that generates the actual constraints: one for the linear and one for the bushy plan space. The high-level algorithm for generating constraint sets does not differ between them.

Note that we have assumed that the number of workers is a power of two and that the number of query tables is a multiple of two for left-deep plans and a multiple of three for bushy plans. Those assumptions simplify our pseudo-code while the extension to the general case (i.e., using only a subset of workers whose cardinality is a power of two) are straight-forward. The number of workers that can be efficiently exploited by our algorithm is however indeed restricted to powers of two and the maximal number of workers is additionally restricted as a function of the query size. We analyze those restrictions in more detail in Section~\ref{analysisSec}.

\begin{algorithm}[t!]
\renewcommand{\algorithmiccomment}[1]{// #1}
\begin{algorithmic}[1]
\State \Comment{Returns pairs of consecutive tables in query $Q$}
\Function{Subsets[Linear]}{$Q$}
\State \Return{$\{\{Q_{2\cdot i},Q_{2\cdot i+1}\}|0\leq i\leq |Q|/2-1\}$}
\EndFunction
\vspace{0.1cm}
\State \Comment{Returns triples of consecutive tables in query $Q$}
\Function{Subsets[Bushy]}{$Q$}
\State \Return{$\{\{Q_{3\cdot i},Q_{3\cdot i+1},Q_{3\cdot i+2}\}|0\leq i\leq|Q|/3-1\}$}
\EndFunction
\vspace{0.1cm}
\State \Comment{Part of power set of $S$ respecting constraints $C$}
\Function{ConstrainedPowerSet[Linear]}{$S,C$}
\State \Return{\Call{Power}{$S$}$\setminus \{\{Q_y\}|(Q_x\prec Q_y)\in C\}$}
\EndFunction
\vspace{0.1cm}
\State \Comment{Part of power set of $S$ respecting constraints $C$}
\Function{ConstrainedPowerSet[Bushy]}{$S,C$}
\State \Return{\Call{Power}{$S$}$\setminus \{\{Q_y,Q_z\}|(Q_x\preceq Q_y|Q_z)\in C\}$}
\EndFunction
\vspace{0.1cm}
\State \Comment{Returns all potential join results (table subsets }
\State \Comment{of query $Q$) that comply with constraints $C$.}
\Function{AdmJoinResults}{$Q,C$}
\State \Comment{Initialize result sets}
\State $R\gets\{\emptyset\}$
\State \Comment{Iterate over subsets of $Q$}
\For{$S\in$\Call{Subsets}{$Q$}}
\State \Comment{Extend join results using Cartesian product}
\State $R\gets R\times $\Call{ConstrainedPowerSet}{$S,C$}
\EndFor
\State \Return{$R$}
\EndFunction
\end{algorithmic}
\caption{Generate all table subsets that comply with the constraints defining a search space partition.\label{admissilbeSetsAlg}}
\end{algorithm}

Constraints restrict the admissible join orders and join trees. We are however ultimately interested in restricting not the number of join orders but rather the number of intermediate results, i.e.\ join result table sets, that can appear in admissible plans. We focus on reducing the number of result table sets as the time and space complexity of the dynamic programming algorithm executed by the workers depends on it. 

We must translate sets of constraints into sets of intermediate results that admissible plans can use. Algorithm~\ref{admissilbeSetsAlg}, more precisely function~\textproc{AdmJoinResults}, accomplishes the translation. The input is the query and a set of constraints. The output is the set of intermediate results that can appear in plans that comply with those constraints. 

Function~\textproc{AdmJoinResults} iterates over all subsets of query tables that constraints can refer to. For left-deep plans those are all pairs of tables with consecutive numbers. For bushy plans those are all triples of consecutive tables. In each iteration of the for loop, the function extends the admissible table sets stored in $R$ by subsets of the table pair (or table triple) considered in the current iteration using a Cartesian product for the extensions. The auxiliary function~\textproc{ConstrainedPowerSet} returns for a given pair (respective triple) or tables all subsets that comply with the constraints. More precisely, if table $Q_x$ needs to be joined before table $Q_y$ in case of left-deep plans then (non-singleton) table sets containing $Q_y$ but not table $Q_x$ do not need to be considered. Equally for bushy plans, if table $Q_x$ must appear before table $Q_y$ when enumerating all table sets containing $Q_z$ then table sets containing $Q_y$ and $Q_z$ but not $Q_x$ are not admissible as join results. 

\begin{example}
Assume that $Q=\{Q_1,Q_2,Q_3,Q_4\}$ and that we have the two constraints $C=\{Q_1\prec Q_2,Q_4\prec Q_3\}$, hence we consider left-deep plans. Then the set of admissible join result sets is generated in function~\textproc{AdmJoinResults} as follows. In the first iteration of the for loop, we extend the elements contained in $R$ (initially this is only the empty set) with the admissible subsets of the first table pair $\{Q_1,Q_2\}$. The admissible subsets are $\{\{\},\{Q_1\},\{Q_1,Q_2\}\}$ and this is at the same time the content of $R$ after the first iteration. The algorithm considers admissible subsets of $\{Q_3,Q_4\}$ in the second iteration (which are the sets $\{\},$ $\{Q_4\},\{Q_3,Q_4\}$) and extends each element with all of the admissible subsets. Hence $R=\{\{\},\{Q_1\},\{Q_1,Q_2\},$ $\{Q_4\},$ $\{Q_1,Q_4\},\{Q_1,Q_2,Q_4\},\{Q_3,Q_4\},\{Q_1,Q_3,Q_4\},$ \\$\{Q_1,Q_2,Q_3,Q_4\}\}$ after the second iteration.
\end{example}

Note that the admissible table sets generated by function~\textproc{AdmJoinResults} do not include all singleton table sets. While all singleton sets must be considered to generate any plan (since we need to select scan plans for each table), singleton sets are treated separately in Algorithm~\ref{WorkerAlg} and it does not matter which of them are included in the result of function~\textproc{AdmJoinResults}. 

\begin{algorithm}[t!]
\renewcommand{\algorithmiccomment}[1]{// #1}
\begin{algorithmic}[1]
\State \Comment{Try all splits of $U\subseteq Q$ into two operands respecting}
\State \Comment{constraints $C$, generate associated plans and prune.}
\Function{TrySplits[Linear]}{$Q,U,C,P$}
\State \Comment{Iterate over potential inner operands}
\For{$u\in U$}
\State \Comment{Check if operand choice satisfies constraints}
\If{$\nexists v\in U:(u\prec v)\in C$}
\State $p\gets$\Call{Join}{$P[U\setminus u],P[u]$}
\State \Call{Prune}{$P,p$}
\EndIf
\EndFor
\EndFunction
\vspace{0.1cm}
\State \Comment{Try all splits of $U\subseteq Q$ into two operands respecting}
\State \Comment{constraints $C$, generate associated plans and prune.}
\Function{TrySplits[Bushy]}{$Q,U,C,P$}
\State \Comment{Determine admissible operands}
\State $A\gets\{\emptyset\}$
\State \Comment{Iterate over set of table triples}
\For{$T\in$\Call{Subsets[Bushy]}{$Q$}}
\State \Comment{Restrict triple to tables in join result}
\State $S\gets T\cap U$
\State \Comment{Form power set of remaining triples}
\State $S\gets$\Call{Power}{$S$}
\State \Comment{Take out sets violating constraints}
\State $S\gets S \setminus \{\{Q_y,Q_z\}|(Q_x\preceq Q_y|Q_z)\in C\}$\label{removePrecLine}
\State \Comment{Remove complement of inadmissible sets}
\State $S\gets S \setminus \{\{Q_x\}|(Q_x\preceq Q_y|Q_z)\in C;Q_y,Q_z\in U\}$\label{removeComplePrecLine}
\State \Comment{Extend admissible splits by Cartesian product}
\State $A\gets A\times S$
\EndFor
\State \Comment{Full set and empty set do not qualify as operands}
\State $A\gets A\setminus\{\emptyset,U\}$
\State \Comment{Iterate over admissible left operands}
\For{$L\in A$}
\State \Comment{Generate plans associated with splits}
\State $p\gets$\Call{Join}{$L,U\setminus L$}
\State \Comment{Discard suboptimal plans}
\State \Call{Prune}{$P,p$}
\EndFor
\EndFunction
\end{algorithmic}
\caption{Generate and prune query plans that correspond to different splits of a join result into two operands.\label{linearSplitsAlg}}
\end{algorithm}

Algorithm~\ref{linearSplitsAlg} shows the function trying out different splits and generating corresponding plans that applies for left-deep plans. This function is called by Algorithm~\ref{WorkerAlg} for each admissible join result. The function iterates over all tables in the join result set $U$ and tries all of them as inner join operands as long as none of the constraints is violated. Plans corresponding to admissible splits are generated and function~\textproc{Prune}, whose pseudo-code we do not specify, compares the newly generated plan against the best plan known so far that produces the same intermediate result as the new one. Only the better one of those two plans remains in the result set. Note that the pruning function can store one optimal plan for each interesting tuple ordering~\cite{Selinger1979}. The pruning function used by the workers might differ from the one used by the master (called~\textproc{FinalPrune} in Algorithm~\ref{masterAlg}) as the tuple ordering is for instance only relevant as long as it can reduce the cost of future operations and does not need to be taken into account anymore for completed plans.

There are actually two mechanisms by which partitioning reduces the time complexity per worker. So far we have focused on the first one: partitioning reduces the time complexity per worker since fewer potential join results need to be considered. An additional advantage of partitioning is however that it allows to reduce the number of splits of join results into two join operands, leading to different query plans that need to be generated and compared. 

The potential for saving computation time by reducing the number of splits is higher for bushy plan spaces since the number of possible splits grows exponentially in the size of the join result. For left-deep plans, the number of splits grows only linearly in the cardinality of the join result as the right join operand is limited to singleton table sets. 

This is why we invest more effort in case of bushy than in case of left-deep plans into properly exploiting the reduction of admissible splits. For left-deep plans, we basically enumerate all possible splits and check whether they comply with the constraints. The complexity of that approach remains linear in the number of possible splits and not in the lower number of admissible splits. The algorithm for bushy plans is more sophisticated as it avoids generating non-admissible splits for bushy plans in the first place. Hence its complexity is linear in the number of admissible rather than possible splits. 

Function~\textproc{TrySplits[Bushy]} generates all admissible splits in a bushy plan space and generates and prunes the associated query plans. The algorithm first generates all admissible join operands and stores them in variable $A$. Each admissible join operand corresponds to the union of one admissible subset for each table triple (constraints are defined on triples of tables). This is why we iterate over all table triples, determine all admissible subsets of the current triple, and combine in each iteration each operand in $A$ with each admissible subset of the current triple (using a similar approach as in Algorithm~\ref{admissilbeSetsAlg}). For a given triple of query tables, we only consider the ones that are included in the join result $U$ that needs to be split. If no constraints are defined on the current triple then the entire power set of the contained table is admissible. Otherwise, we must remove subsets violating the precedence constraints (line~\ref{removePrecLine}) but we must also remove subsets whose complement (in the contained triple tables) would violate the precedence constraints (line~\ref{removeComplePrecLine}) as the second join operand is the complement of the first operand.

Having determined all admissible join operands (whose complement is admissible, too), we iterate over all of them, generate plans and discard sub-optimal plans.

\section{Complexity Analysis}
\label{analysisSec}

We analyze the asymptotic complexity of the algorithm presented in the previous section according to multiple metrics: we analyze the asymptotic amount of data sent over the network in Section~\ref{communicationSub}, the consumed amount of main memory in Section~\ref{memorySub}, and the execution time in Section~\ref{timeSub}.

We simplify the following analysis by assuming that only one scan and join operator is used. We also assume that only one cost metric is considered when comparing query plans and that no interesting orders are present. In Section~\ref{extensionsSub}, we discuss how the analysis can be extended. In Section~\ref{optPartitioningSub} we discuss the question of whether our partitioning methods can be improved and show that they are optimal at least within a restricted space of partitioning methods. 

We introduce notations that are used throughout this section. We denote by $n=|Q|$ the number of query tables to join and by $m$ the number of worker machines. We assume that $m\leq 2^{\lfloor n/2\rfloor}$ for linear plan search spaces and $m\leq2^{\lfloor n/3\rfloor}$ for bushy plan spaces. This is required as we assume that the table sets that different constraints refer to are mutually disjunct. We use two tables per constraint for linear plan spaces and three tables for bushy plan spaces. We denote by $l=\lfloor\log_2(m)\rfloor$ the number of constraints per plan space partition. By $b_q$ we designate the byte size of the input query. By $b_p$ we denote the byte size of a corresponding query plan.

\subsection{Network Communication}
\label{communicationSub}

We analyze the asymptotic size of the data that is sent over the network during optimization of one query.

\begin{theorem}
The amount of data sent over the network is in $O(m\cdot(b_q+b_p))$.
\end{theorem}
\begin{proof}
Different workers do not communicate with each other so data is only sent between master and workers. This happens at two occasions: when assigning each worker to a plan space partition and when collecting the best plans for each partition. The input for each worker is the query (with space consumption $b_q$) and two integer numbers with constant space consumption. We consider one plan cost metric and no interesting orders (while the extensions are discussed later). The output of each worker is therefore a single query plan with space consumption $b_q$. 
\end{proof}

\subsection{Main Memory}
\label{memorySub}

We analyze the amount of main memory that each worker requires during optimization. Note that the main memory consumption of the master is negligible as it delegates optimization. The main memory consumed per worker node depends on the number of admissible join results.

\begin{theorem}
Each linear plan space partition restricted by $l$ constraints has $O(2^n\cdot (3/4)^l)$ admissible join results.\label{relevantLinearTheorem}
\end{theorem}
\begin{proof}
The proof is an induction over the number of constraints $l$. For $l=0$ (induction start), all subsets of $Q$ are admissible and their number is in $O(2^n)$. Assume the induction holds up to $L$ constraints. We will see that it holds for $L+1$ constraints as well. All constraints refer to different tables. Hence the first $L$ constraints do not influence the occurrence frequency of the two tables $x$ and $y$ that the $L+1$-th constraint refers to. More precisely, among the table sets that remain admissible after considering the first $L$ constraints, the fraction of table sets containing $x$ and $y$, $x$ but not $y$, $y$ but not $x$, and neither $x$ nor $y$, is always $1/4$. Denote by $x\prec y$ the $L+1$-th constraint stating that we must join $x$ before $y$. Then join results containing $y$ but not $x$ are inadmissible, the number of admissible table sets is reduced by factor $3/4$, and the induction holds.
\end{proof}

\begin{theorem}
Each bushy plan space partition restricted by $l$ constraints has $O(2^n\cdot (7/8)^l)$ admissible join results.\label{relevantBushyTheorem}
\end{theorem}
\begin{proof}[Sketch]
The proof follows closely the one of Theorem~\ref{relevantLinearTheorem} with the difference that each constraint of the form $x\preceq y|z$ excludes all table sets that contain $y$ and $z$ but not $x$ and their fraction is always $1/8$ among the table sets satisfying all other constraints.
\end{proof}

We use the number of admissible join results to calculate main memory consumption.

\begin{theorem}
The main memory consumption per node is in $O(2^n\cdot(3/4)^l)$ for linear plan spaces and $O(2^n\cdot(7/8)^l)$ for bushy plan spaces.\label{memoryLinearTheorem}
\end{theorem}
\begin{proof}
The main memory consumption per worker dominates the consumption of the master. The variable with dominant space consumption are the ones storing admissible join results (variable $joinRes$ in Algorithm~\ref{WorkerAlg}) and the one assigning table sets to optimal plans (variable $P$). We currently assume one plan cost metric and therefore only one plan is optimal per table set. Storing plans generally takes $O(n)$ space but here each plan can be represented by at most two pointers to optimal sub-plans stored for table subsets which requires only $O(1)$ space. The total main memory consumption follows from Theorems~\ref{relevantLinearTheorem} and \ref{relevantBushyTheorem}.
\end{proof}

\subsection{Execution Time}
\label{timeSub}

We analyze time complexity. Note that the pseudo-code presented in Section~\ref{algSec} is rather abstract and does not contain certain steps that are crucial for efficiency: as we mentioned in Section~\ref{algSec}, we assume for instance that constraints are indexed such that we can find all constraints in which a given table appears in constant time. For the following analysis, we assume that such commonsense optimizations have been applied (we use them as well in our implementation that is evaluated in Section~\ref{experimentsSec}). 

We first analyze execution time on the master. 

\begin{theorem}
The master requires $O(m\cdot(b_q+b_p))$ time.
\end{theorem}
\begin{proof}
The master distributes the query and the partition ID to all $m$ workers. Assuming that the required time is proportional to the amount of data being sent, distributing tasks takes $O(mb_q)$ time and collecting plans from the workers is in $O(mb_p)$. After receiving all plans, the master iterates over the $m$ plans that were returned from the workers (and whose cost was already calculated) and determines the one with minimal cost. This has complexity $O(m)$. 
\end{proof}

Now we analyze the time complexity of processing a linear plan space partition.

\begin{theorem}
The time complexity for processing a linear plan space partition by one of the workers is $O(n\cdot 2^n\cdot(3/4)^l)$.
\end{theorem}
\begin{proof}
A worker performs three main steps per invocation: translating the partition ID into constraints, translating constraints into admissible join result sets, and determining the optimal plan among the plans using only admissible join results. The operation with dominant time complexity is the determination of the optimal plan. For each admissible join result set, we iterate over less than $n$ inner join operands. The number of admissible join result sets is in $O(2^n\cdot(3/4)^l)$ according to Theorem~\ref{relevantLinearTheorem}. Generating a plan from two sub-plans, calculating its cost via recursive formulas, and comparing it with the best previously generated plan joining the same tables requires only constant time.
\end{proof}

\begin{theorem}
The time complexity for processing a bushy plan space partition by one of the workers is $O(3^n\cdot(21/27)^l)$.
\end{theorem}
\begin{proof}
Determining the optimal plan in the restricted plan space partition is the operation with dominant time complexity. The time complexity for finding an optimal plan is lower-bounded by the number of considered result table sets. It is proportional to the number of considered join operand pairs. 

For each table there are in general three possibilities for how it appears in a pair of join operands: either it appears in the left operand or in the right operand or it does not appear (neither in the operands nor in the join result). Join operands are constructed from admissible subsets of table triples. If no constraint is defined on a given triple then all splits are admissible which makes $3^3=27$ possible pairs. If a constraint is defined on a triple then some of those 27 possibilities are not admissible. If the constraint is $x\preceq y|z$ then the following six splits of triple $\{x,y,z\}$ are excluded: all splits whose union contains $y$ and $z$ but not $x$ (this applies to four splits) and all splits that assign $y$ and $z$ to one operand and $x$ to the other one (this applies to two splits). The ratio of admissible to possible splits is therefore $21/27$ for each triple with a constraint on it. 
\end{proof}

As the time complexity of the worker processes dominates the complexity of the master process and as all workers execute in parallel, the time complexity of a single worker is the complexity of the entire optimization process.

\subsection{Extensions}
\label{extensionsSub}

So far we considered one plan cost metric, no interesting orders, and no alternative operator implementations. It is however straight-forward to extend the analysis as we sketch out in the following. 

Considering multiple alternative operator implementations for scan and join operations influences only time complexity. Time complexity grows linearly in the number of operators as each join operator implementation must be considered for each possible pair of join operands. Annotating the operations within query plans by an operator ID does neither change the asymptotic space complexity in main memory nor the asymptotic communication overhead as storing an integer ID requires constant space. 

Considering interesting orders or considering multiple plan cost metrics both have the effect that multiple plans can be optimal for joining the same set of tables. The number of interesting orders restricts the number of plans that need to be stored per table set. Assuming that multiple plan cost metrics are considered while using an approximation factor, the number of Pareto-optimal plans per table set can be bounded as shown by Trummer and Koch~\cite{Trummer2014}. The number of plans sent from workers to master, and therefore the communication overhead, increases linearly in the number of plans stored per table set. Main memory consumption also increases linearly in the number of plans. Time complexity increases proportional to the cube of the number of plans for the following reason: when searching for the optimal plan within each plan space partition, we need to consider all pairs of optimal plans for each split of a table set into two join operands~\cite{Trummer2014}. This accounts for a quadratic increase in complexity. Additionally, pruning takes longer as we need to compare plans not against one but multiple optimal plans. Together this implies a cubic increase in complexity. Note that plans need to be compared according to multiple cost metrics but the number of plan cost metrics is usually considered a constant~\cite{Trummer2014, Trummer2015}.

\subsection{Optimality of Partitioning}
\label{optPartitioningSub}

Execution time and main memory consumption both depend on the number of intermediate join results that need to be treated by each worker. With our partitioning scheme, the number of join results per worker reduces by factor 3/4 in case of linear plan spaces and by factor 7/8 for bushy plans, each time that the number of workers doubles. As the ideal factor of 1/2 is not reached there must be join results that are assigned to multiple workers. This raises the question of whether we can do better and reduce the number of intermediate results per worker node by a lower factor. 

We answer that question in the following for partitioning methods that are similar to the one we apply. Similar methods are methods that divide the power set of query tables into subsets based on which out of two, respective three, fixed tables are present. Each of the resulting subsets is assigned to part of the workers and each worker generates all plans whose join results are contained in its assigned subsets (each worker constructs scan plans for all single tables, independently from the assigned join results). Workers do not exchange partial plans and hence must generate completed plans and start optimization from scratch. 

The following theorems study the best speedup that is achievable by partitioning the plan space between two workers. The reasoning can however be generalized to higher degrees of parallelism.

\begin{theorem}
Doubling the number of workers cannot reduce the maximal number of join results per worker by less than factor 3/4 in linear plan spaces.
\end{theorem}
\begin{proof}
For a fixed pair of tables $\{x,y\}$ out of all query tables, we denote by $\overline{x}y$ the set of table sets containing $y$ but not $x$, by $xy$ the sets containing both tables, by $\overline{xy}$ sets containing neither $x$ nor $y$, and by $x\overline{y}$ the remaining sets. Each worker must obtain subset $xy$ in order to generate complete plans. The cardinality of the set of joined tables can only increase by one from one join to the next in a left-deep plan space. Each worker needs therefore either join results from $\overline{x}y$ or from $x\overline{y}$ in order to generate any valid plan. Set $\overline{xy}$ must be assigned to at least one worker since the plan space partitioning is otherwise incomplete. 
\end{proof}

\begin{theorem}
Doubling the number of workers cannot reduce the maximal number of join results per worker by less than factor 7/8 in bushy plan spaces.
\end{theorem}
\begin{proof}[Sketch]
For a triple of tables $\{x,y,z\}$, we use a similar notation as before to characterize join result sets and denote for instance by $x\overline{y}z$ all sets containing $x$ and $z$ but not $y$. Both workers require $xyz$ for the same reason as before. Assume that we do not assign the set $\overline{xyz}$ to both workers. The worker to which $\overline{xyz}$ is assigned is the only worker that can consider plans joining the other tables besides $x$, $y$, and $z$ independently before joining with the triple tables. This means that this worker needs to cover all possible join orders for $x$, $y$, and $z$. Hence it requires all join result sets which defeats the purpose of partitioning. 

Assume now that we do not assign the set $x\overline{yz}$ to the first worker. Then the second worker is the only one that can consider plans of the form $(x\Join\ldots)\Join(y\Join\ldots)$ and hence requires $\overline{x}y\overline{z}$ and by analogue reasoning also $\overline{xy}z$ in addition to $x\overline{yz}$ in order to make sure that the whole plan space is covered. As the second worker is at the same time the only one that can consider plans of the form $((x\Join\ldots)\Join y)\Join\ldots$, it requires at the same time $xy\overline{z}$ and $x\overline{y}z$. Since only the second worker can treat plans of the form $(x\Join\ldots)\Join(y\Join z)$, it requires also $\overline{x}yz$. So the second worker obtains at least 7 sets of join results. The same happens when not assigning $\overline{x}y\overline{z}$ or $\overline{xy}z$ to the first worker. We have the option of not assigning one of the three sets containing two out of the three tables $\{x,y,z\}$ to the first worker in which case we need to assign the other two to the second worker. The maximal number of intermediate result splits per worker remains 7/8.
\end{proof}

\section{Experimental Evaluation}
\label{experimentsSec}



We evaluate the scalability of our query optimization algorithm on a large cluster with 100 nodes and parallelize over up to 256 Spark executors. Parallelizing query optimization on clusters is useful if query plans are also executed on a cluster: it is preferable to use all available resources for optimization instead of leaving nodes idle until serial optimization finishes. For queries that are executed regularly, a cluster can be used before run time to calculate optimal query plans if the search space is too large for optimization on a single node. While parallelizing query optimization on a cluster is hence a relevant application scenario, we also selected it specifically because it is a very challenging scenario for parallelization due to high communication cost and setup overhead. The fact that our algorithm scales even on a cluster provides in our opinion strong evidence for that it scales in a multitude of other scenarios as well. Our algorithm is not restricted to shared-nothing architectures but can also be applied in shared-memory settings. 

We evaluate our algorithm, in comparison with a baseline, for traditional query optimization with one plan cost metric as well as for multi-objective query optimization~\cite{Trummer2014} where query plans are compared according to multiple cost metrics. We also calculate the speedups that we obtain by parallelization compared to serial algorithms~\cite{Selinger1979, Trummer2014}. Section~\ref{setupSub} describes our experimental setup and Section~\ref{resultsSub} our experimental results. 

\subsection{Experimental Setup}
\label{setupSub}

We evaluate our algorithm on a cluster with 100 nodes. Each node is equipped with two Intel Xeon E5-2630 v2 CPUs featuring six cores each running at 2.60GHz; 128~GB of main memory and 20~TB of hard disk capacity are available per node. The cluster runs Ubuntu Linux, version~14.04. 

All benchmarked algorithms use Spark~1.5 on Yarn~2.7.1 and are implemented in Java~1.7. We implemented the algorithm from Section~\ref{algSec} and abbreviate it by MPQ (for massively parallel query optimization) in the following plots. We compare against an algorithm that is representative for the rather fine-grained approaches to parallelizing query optimization proposed so far. They were targeted at shared-memory architectures and moderate degrees of parallelism~\cite{Han2008, Han2009}. We call that algorithm SMA (for shared-memory approach) in the following plots. In this algorithm, the master node assigns to each worker a set of join results for which to calculate optimal plans based on the optimal plans that were generated by other workers. This means that intermediate results need to be shared between workers and that the master needs to assign multiple rounds of tasks to the workers. The comparison between MPQ and SMA is of course unfair as both were developed for different architectures and different degrees of parallelism. We are however not aware of any other query optimization algorithms targeted at shared-nothing architectures. 

We use up to 256 Spark executors and reserve up to 40~GB of main memory per executor (query optimization requires large amounts of memory, in particular in case of multiple plan cost metrics~\cite{Trummer2014}). We set the maximum message size to 1~GB (MSA needs to send large messages).

We compare algorithms in linear and bushy plan spaces. We always consider the full plan space and do not heuristically restrict the use of cross products as this might miss optimal plans~\cite{Ono1990}. As we allow cross products, the number of intermediate results to consider and hence performance of our optimization algorithms does not critically depend on the structure of the query join graphs. We generate random star join graphs, table cardinalities, and predicate selectivity values by the method introduced by Steinbrunn et al.~\cite{Steinbrunn1997} which is commonly used for query optimization benchmarks~\cite{Bruno, Trummer2015a, Trummer2015}. In a first series of experiments, we consider execution time as only cost metric and use standard cost formulas~\cite{Steinbrunn1997} to estimate the cost of standard join operators such as block-nested loop join, hash join, and sort-merge join. In a second series of experiments, we consider two plan cost metrics and the goal is hence to approximate the set of Pareto-optimal plans (a plan is Pareto-optimal if no other plan has better cost according to all cost metrics~\cite{Trummer2014}). Our second cost metric (in addition to execution time) is the buffer space consumption such that we calculate optimal cost tradeoffs between execution time and buffer space consumption. Note that those are two cost metrics that are frequently used for benchmarking multi-objective query optimization algorithms~\cite{Trummer2014, Trummer2015a}. 

For the series of experiments with two plan cost metrics, we replace the standard pruning function in our algorithms by a pruning function that was used in prior work for multi-objective query optimization with formal near-optimality guarantees~\cite{Trummer2014, Trummer2015a}. That pruning function is parameterized by an approximation factor $\alpha$ and we set $\alpha=10$.

\subsection{Experimental Results}
\label{resultsSub}

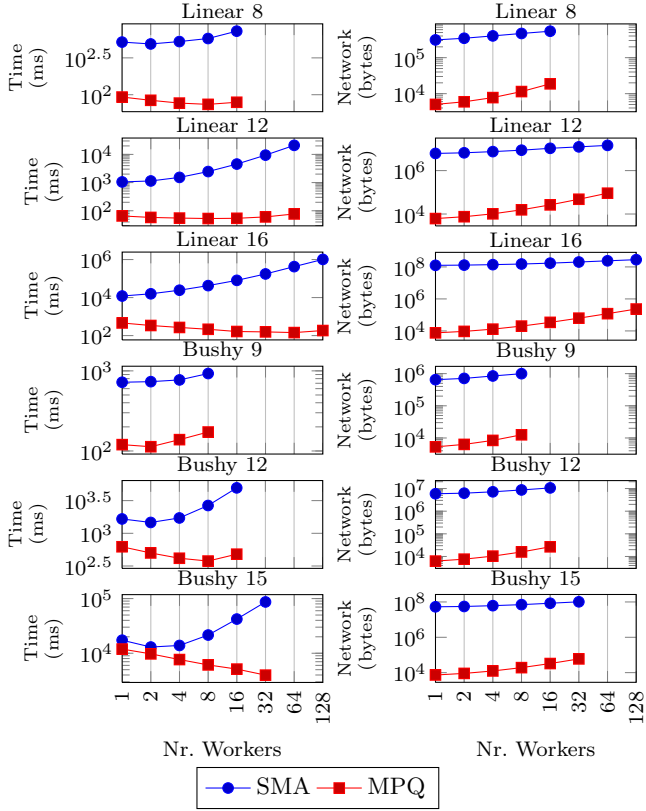
\begin{figure}[t!]
\centering
\begin{tikzpicture}
\begin{groupplot}[group style={group size=2 by 6, 
x descriptions at=edge bottom, horizontal sep=1.5cm, vertical sep=0.35cm},
width=4.25cm, height=2.75cm,
xlabel=Nr.\ Workers,
xlabel near ticks, ylabel near ticks, 
x label style={font=\small}, y label style={font=\scriptsize, align=center},
xmajorgrids, xtick={0,1,2,3,4,5,6,7,8}, xticklabels={1,2,4,8,16,32,64,128},
x tick label style={font=\small, rotate=90}, y tick label style={font=\small},
title style={yshift=-0.25cm},
xmin=0,xmax=7,
legend to name=soqoComparisonLegend, legend style={font=\footnotesize},legend columns=4,
legend entries={SMA, MPQ},
ymode=log, title style={font=\small}]
\nextgroupplot[ylabel={Time\\(ms)},title=Linear 8]
\addplot table[x index=0,y index=1] {plotsdata/soqo/comparison/linear8_masterTime};
\addplot table[x index=0,y index=2] {plotsdata/soqo/comparison/linear8_masterTime};
\nextgroupplot[ylabel={Network\\(bytes)},title=Linear 8]
\addplot table[x index=0,y index=1] {plotsdata/soqo/comparison/linear8_network};
\addplot table[x index=0,y index=2] {plotsdata/soqo/comparison/linear8_network};
\nextgroupplot[ylabel={Time\\(ms)},title=Linear 12]
\addplot table[x index=0,y index=1] {plotsdata/soqo/comparison/linear12_masterTime};
\addplot table[x index=0,y index=2] {plotsdata/soqo/comparison/linear12_masterTime};
\nextgroupplot[ylabel={Network\\(bytes)},title=Linear 12]
\addplot table[x index=0,y index=1] {plotsdata/soqo/comparison/linear12_network};
\addplot table[x index=0,y index=2] {plotsdata/soqo/comparison/linear12_network};
\nextgroupplot[ylabel={Time\\(ms)},title=Linear 16]
\addplot table[x index=0,y index=1] {plotsdata/soqo/comparison/linear16_masterTime};
\addplot table[x index=0,y index=2] {plotsdata/soqo/comparison/linear16_masterTime};
\nextgroupplot[ylabel={Network\\(bytes)},title=Linear 16]
\addplot table[x index=0,y index=1] {plotsdata/soqo/comparison/linear16_network};
\addplot table[x index=0,y index=2] {plotsdata/soqo/comparison/linear16_network};
\nextgroupplot[ylabel={Time\\(ms)},title=Bushy 9]
\addplot table[x index=0,y index=1] {plotsdata/soqo/comparison/bushy9_masterTime};
\addplot table[x index=0,y index=2] {plotsdata/soqo/comparison/bushy9_masterTime};
\nextgroupplot[ylabel={Network\\(bytes)},title=Bushy 9]
\addplot table[x index=0,y index=1] {plotsdata/soqo/comparison/bushy9_network};
\addplot table[x index=0,y index=2] {plotsdata/soqo/comparison/bushy9_network};
\nextgroupplot[ylabel={Time\\(ms)},title=Bushy 12]
\addplot table[x index=0,y index=1] {plotsdata/soqo/comparison/bushy12_masterTime};
\addplot table[x index=0,y index=2] {plotsdata/soqo/comparison/bushy12_masterTime};
\nextgroupplot[ylabel={Network\\(bytes)},title=Bushy 12]
\addplot table[x index=0,y index=1] {plotsdata/soqo/comparison/bushy12_network};
\addplot table[x index=0,y index=2] {plotsdata/soqo/comparison/bushy12_network};
\nextgroupplot[ylabel={Time\\(ms)},title=Bushy 15]
\addplot table[x index=0,y index=1] {plotsdata/soqo/comparison/bushy15_masterTime};
\addplot table[x index=0,y index=2] {plotsdata/soqo/comparison/bushy15_masterTime};
\nextgroupplot[ylabel={Network\\(bytes)},title=Bushy 15]
\addplot table[x index=0,y index=1] {plotsdata/soqo/comparison/bushy15_network};
\addplot table[x index=0,y index=2] {plotsdata/soqo/comparison/bushy15_network};
\end{groupplot}
\end{tikzpicture}

\ref{soqoComparisonLegend}
\caption{MPQ outperforms MSA by up to four orders of magnitude in terms of optimization time; scalability of MPQ is limited due to the query sizes.\label{soqoComparisonFig}}
\end{figure}

Due to space restrictions, we show only an extract of our full experimental results. The presented results are however representative and we observed the same tendencies in our additional experiments.

We start by discussing the results for traditional query optimization with one plan cost metric. Figure~\ref{soqoComparisonFig} shows a comparison between MPQ and MSA in terms of optimization time and in terms of the amount of data exchanged between cluster nodes. Each data point in the plots corresponds to the median of the results for twenty randomly generated queries. We compare algorithms for different plan spaces (linear and bushy) and for different query sizes (number of joined tables). We try different degrees of parallelism for each plan space, adapting the maximal parallelism to the search space size (we scaled up to the maximal degree of parallelism that MPQ can exploit based on the number of disjoint table pairs or triples) up to a maximum of 128 workers. We try smaller query sizes for the bushy plan space than for the linear plan space as the size of the bushy search space grows faster in the number of query tables. Note that we also consider Cartesian product joins in contrast to prior evaluations of parallel query optimization algorithms. This makes the plan space much larger for the same number of tables. Still the search spaces treated in Figure~\ref{soqoComparisonFig} are of moderate size and we try larger search spaces in the following.

MPQ outperforms MSA by up to four orders of magnitude in optimization time. The reason is the large amount of data that MSA has to send over the network, due to the need for sharing intermediate results between workers, and the overheads on the master node by fine-grained task management. The amount of data sent by MSA reaches several hundreds of megabytes while our algorithm sends at most 234 kilobytes and in most cases significantly less than that. As outlined before, MSA is not designed for shared-nothing scenarios and the performance gap between the algorithms is to be expected. 

The search space sizes in Figure~\ref{soqoComparisonFig} represent approximately the limit of what the competitor algorithm MSA can treat within reasonable amounts of time. For our MPQ algorithm, the considered search spaces are actually too small to justify parallelization. This is why we see in most plots in Figure~\ref{soqoComparisonFig} no decrease in optimization time for MPQ with growing degree of parallelism (except for the bushy search space with 15 query tables). The absolute optimization times are for MPQ already very low even for a single worker so parallelization is not needed yet. The amount of network traffic and the management overhead increase for both algorithms once the number of workers increases. MSA can only benefit in few cases from parallelization and only up to a degree of parallelism of four. 

The computation time of MSA increases quickly in the query size and in the degree of parallelism as well (reaching more than 15 minutes per test case for 16-table joins). This is why we exclude it from the following series of experiments. 

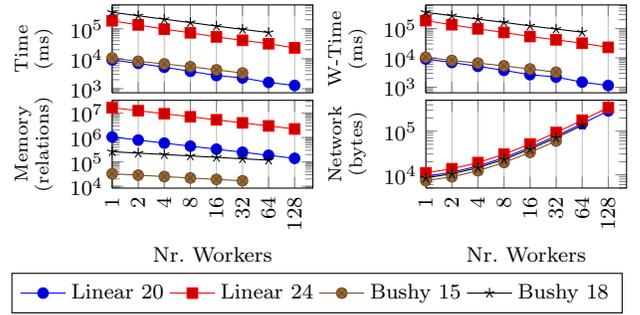
\begin{figure}[t!]
\centering
\begin{tikzpicture}
\begin{groupplot}[group style={group size=2 by 2, 
x descriptions at=edge bottom, horizontal sep=1.5cm, vertical sep=0.1cm},
width=4.25cm, height=2.75cm,
xlabel=Nr.\ Workers,
xlabel near ticks, ylabel near ticks, 
x label style={font=\small}, y label style={font=\scriptsize, align=center},
xmajorgrids, xtick=data, xticklabels={-1,1,2,4,8,16,32,64,128},
x tick label style={font=\small, rotate=90}, y tick label style={font=\small},
title style={yshift=-0.25cm},
xmin=0,
legend to name=soqoScalabilityLegend, legend style={font=\small},legend columns=4,
legend entries={Linear 20,Linear 24,Bushy 15,Bushy 18}]
\nextgroupplot[ymode=log,ylabel={Time\\(ms)}]
\addplot table[x index=0,y index=1] {plotsdata/soqo/scalability/linear20_masterTime};
\addplot table[x index=0,y index=1] {plotsdata/soqo/scalability/linear24_masterTime};
\addplot table[x index=0,y index=1] {plotsdata/soqo/scalability/bushy15_masterTime};
\addplot table[x index=0,y index=1] {plotsdata/soqo/scalability/bushy18_masterTime};
\nextgroupplot[ymode=log,ylabel={W-Time\\(ms)}]
\addplot table[x index=0,y index=1] {plotsdata/soqo/scalability/linear20_workerTime};
\addplot table[x index=0,y index=1] {plotsdata/soqo/scalability/linear24_workerTime};
\addplot table[x index=0,y index=1] {plotsdata/soqo/scalability/bushy15_workerTime};
\addplot table[x index=0,y index=1] {plotsdata/soqo/scalability/bushy18_workerTime};
\nextgroupplot[ymode=log,ylabel={Memory\\(relations)}]
\addplot table[x index=0,y index=1] {plotsdata/soqo/scalability/linear20_memory};
\addplot table[x index=0,y index=1] {plotsdata/soqo/scalability/linear24_memory};
\addplot table[x index=0,y index=1] {plotsdata/soqo/scalability/bushy15_memory};
\addplot table[x index=0,y index=1] {plotsdata/soqo/scalability/bushy18_memory};
\nextgroupplot[ymode=log,ylabel={Network\\(bytes)}]
\addplot table[x index=0,y index=1] {plotsdata/soqo/scalability/linear20_network};
\addplot table[x index=0,y index=1] {plotsdata/soqo/scalability/linear24_network};
\addplot table[x index=0,y index=1] {plotsdata/soqo/scalability/bushy15_network};
\addplot table[x index=0,y index=1] {plotsdata/soqo/scalability/bushy18_network};
\end{groupplot}
\end{tikzpicture}

\ref{soqoScalabilityLegend}
\caption{MPQ scales steadily for sufficiently large search spaces and one plan cost metric.\label{soqoScalabilityFig}}
\end{figure}

Figure~\ref{soqoScalabilityFig} shows results for larger search spaces and only for MPQ. The figure shows total optimization time (measured on the master node) as well as the maximal optimization time measured over all workers (``W-Time'' in the figure). The fact that the difference between both is small indicates that the management overhead on the master node is negligible. We show network traffic and additionally the maximal main memory consumption over all of the workers (the master does not perform optimization itself and its main memory consumption is negligible). We scale for each query size up to the maximal degree of parallelism supported by our algorithm (determined by the number of table pairs for linear plans and the number of table triples for bushy plans) and maximally up to 128 workers. 

As search space sizes are large enough, we see steady scaling for all degrees of parallelism that are theoretically supported by our algorithm without diminishing returns for higher number of workers. The scaling is slightly better for linear plans than for bushy plans which matches precisely our theoretical predictions from Section~\ref{analysisSec} (execution time decreases by factor 3/4 for linear plans but only by factor 21/27 for bushy plans, each time that the degree of parallelism doubles). Unlike for MSA, the network traffic created by MPQ depends only marginally on the query size as no intermediate results have to be exchanged between workers or between workers and master. Only the query itself and the final plan generated by each worker are sent. The maximal main memory consumption on the workers (measured by the number of relations for which to store optimal plans) equally decreases steadily with increasing parallelization. Here the decrease for bushy plans is slower than for linear plans which again matches our theoretical results. 

If we use one worker then no constraints on the join order are defined. Then MPQ is equivalent to the classical Selinger algorithm~\cite{Selinger1979} as it treats the same table sets in the same order. Hence we compare the optimization time when executing our algorithm on a single worker (not measuring master computation time and communication overheads) to the optimization time of the parallel version (including master computation time and communication overheads) to obtain the speedup of our algorithm compared to serial query optimization. With 128 workers, we obtain for left-deep plans a speedup of 8.1 for 24 query tables and a speedup of 7.2 for 20 tables. With 32 workers we have a speedup of 3.2 for 15-table joins and bushy query plans and a speedup of 4.8 for 18-table joins and 64 workers.



We finally want to point out that our Java-based implementation is not optimized for maximum efficiency. It is rather optimized for modularity, allowing to ``plug-in'' different search spaces and cost metrics. This enables us to execute experiments over a broad range of scenarios but it also introduces overheads in some of the functions that are most frequently called during optimization. We believe that optimization efficiency can be significantly improved by specializing the algorithm to a single scenario. 

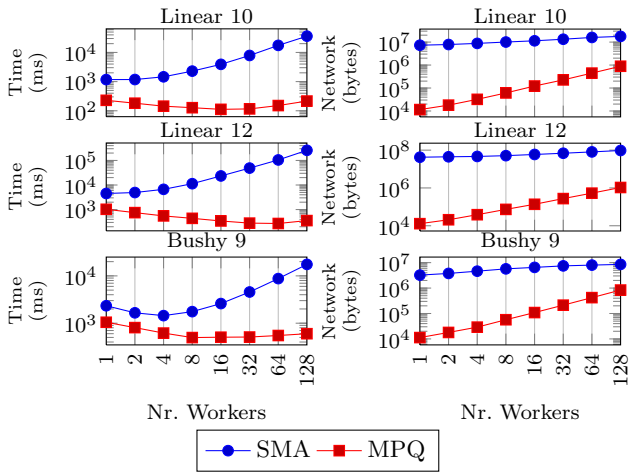
\begin{figure}[t!]
\centering
\begin{tikzpicture}
\begin{groupplot}[group style={group size=2 by 3, 
x descriptions at=edge bottom, horizontal sep=1.5cm, vertical sep=0.35cm},
width=4.25cm, height=2.75cm,
xlabel=Nr.\ Workers,
xlabel near ticks, ylabel near ticks, 
x label style={font=\small}, y label style={font=\scriptsize, align=center},
xmajorgrids, xtick={0,1,2,3,4,5,6,7,8}, xticklabels={1,2,4,8,16,32,64,128,256},
x tick label style={font=\small, rotate=90}, y tick label style={font=\small},
title style={yshift=-0.25cm},
xmin=0,xmax=7,
legend to name=moqoComparisonLegend, legend style={font=\footnotesize},legend columns=4,
legend entries={SMA, MPQ},
ymode=log, title style={font=\small}]
\nextgroupplot[ylabel={Time\\(ms)},title=Linear 10]
\addplot table[x index=0,y index=1] {plotsdata/moqo/comparison/linear10_masterTime};
\addplot table[x index=0,y index=2] {plotsdata/moqo/comparison/linear10_masterTime};
\nextgroupplot[ylabel={Network\\(bytes)},title=Linear 10]
\addplot table[x index=0,y index=1] {plotsdata/moqo/comparison/linear10_network};
\addplot table[x index=0,y index=2] {plotsdata/moqo/comparison/linear10_network};
\nextgroupplot[ylabel={Time\\(ms)},title=Linear 12]
\addplot table[x index=0,y index=1] {plotsdata/moqo/comparison/linear12_masterTime};
\addplot table[x index=0,y index=2] {plotsdata/moqo/comparison/linear12_masterTime};
\nextgroupplot[ylabel={Network\\(bytes)},title=Linear 12]
\addplot table[x index=0,y index=1] {plotsdata/moqo/comparison/linear12_network};
\addplot table[x index=0,y index=2] {plotsdata/moqo/comparison/linear12_network};
\nextgroupplot[ylabel={Time\\(ms)},title=Bushy 9]
\addplot table[x index=0,y index=1] {plotsdata/moqo/comparison/bushy9_masterTime};
\addplot table[x index=0,y index=2] {plotsdata/moqo/comparison/bushy9_masterTime};
\nextgroupplot[ylabel={Network\\(bytes)},title=Bushy 9]
\addplot table[x index=0,y index=1] {plotsdata/moqo/comparison/bushy9_network};
\addplot table[x index=0,y index=2] {plotsdata/moqo/comparison/bushy9_network};
\end{groupplot}
\end{tikzpicture}

\ref{moqoComparisonLegend}
\caption{MPQ outperforms MSA but its scalability is limited by small query sizes.\label{moqoComparisonFig}}
\end{figure}

We discuss the results for multi-objective query optimization. Figure~\ref{moqoComparisonFig} shows a comparison between multi-objective versions of MSA and MPQ (both algorithms use the same pruning function that we reconfigured to consider two cost metrics). The tendencies are similar as for single-objective query optimization. Optimization times and network traffic are significantly lower for MPQ than for MSA. The network traffic of MPQ has however increased when comparing to the results for single-objective query optimization. The reason is that each worker must now send the set of all Pareto-optimal plans in its respective plan space partition back to the master instead of only one plan. The median number of complete Pareto-optimal plans per query was 21 for 12-table joins  when considering left-deep plans and 16 for 9-table joins in a bushy plan space. 

Instead of exploiting a high degree of parallelism, MSA suffers significantly once the number of workers increases due to network traffic and coordination overhead. The maximal degree of parallelism that was beneficial to MSA is eight. This is also the number of threads that prior algorithms were maximally evaluated on. MPQ benefits from parallelism up to 32 workers for 10-table joins and left-deep plans, for up to 64 workers for 12-table joins, and for up to eight workers for 9-table joins and bushy plan spaces which corresponds to the number of disjoint table pairs respective triples. The absolute run times of MPQ are however so low that parallelization is unnecessary. 

Figure~\ref{linearMoqoScalabilityFig} shows results for MPQ on queries that are sufficiently large to exploit large degrees of parallelism. The scaling is steady and without noticeable diminishing returns effects up to the maximal number of 256 workers. Note that the run times of MPQ in Figure~\ref{linearMoqoScalabilityFig} are lower than the run times of MSA in Figure~\ref{moqoComparisonFig}, even though we consider significantly larger search spaces in Figure~\ref{linearMoqoScalabilityFig}. We tested scalability for bushy plans and more than 9 query tables and saw steady scaling up to the number of table triples in the query. We omit those results due to space restrictions. 

Our algorithm is for one worker equivalent to a classical algorithm for multi-objective query optimization~\cite{Trummer2014}. We calculate speedups in a similar way as before and obtain a speedup of 5.1 for 16-table joins, 5.5 for 18-table joins, and 9.4 for 20-table joins.

\begin{figure}[t!]
\centering
\begin{tikzpicture}
\begin{groupplot}[group style={group size=2 by 2, 
x descriptions at=edge bottom, horizontal sep=1.5cm, vertical sep=0.1cm},
width=4.25cm, height=2.75cm,
xlabel=Nr.\ Workers,
xlabel near ticks, ylabel near ticks, 
x label style={font=\small}, y label style={font=\scriptsize, align=center},
xmajorgrids, xtick={0,1,2,3,4}, xticklabels={16,32,64,128,256},
x tick label style={font=\small, rotate=90}, y tick label style={font=\small},
title style={yshift=-0.25cm},
xmin=0,
legend to name=linearMoqoScalabilityLegend, legend style={font=\small},legend columns=4,
legend entries={Linear 16, Linear 18, Linear 20}]
\nextgroupplot[ymode=log,ylabel={Time\\(ms)}]
\addplot table[x index=0,y index=1] {plotsdata/moqo/scalability/linear16_masterTime};
\addplot table[x index=0,y index=1] {plotsdata/moqo/scalability/linear18_masterTime};
\addplot table[x index=0,y index=1] {plotsdata/moqo/scalability/linear20_masterTime};
\nextgroupplot[ymode=log,ylabel={W-Time\\(ms)}]
\addplot table[x index=0,y index=1] {plotsdata/moqo/scalability/linear16_workerTime};
\addplot table[x index=0,y index=1] {plotsdata/moqo/scalability/linear18_workerTime};
\addplot table[x index=0,y index=1] {plotsdata/moqo/scalability/linear20_workerTime};
\nextgroupplot[ymode=log,ylabel={Memory\\(relations)}]
\addplot table[x index=0,y index=1] {plotsdata/moqo/scalability/linear16_memory};
\addplot table[x index=0,y index=1] {plotsdata/moqo/scalability/linear18_memory};
\addplot table[x index=0,y index=1] {plotsdata/moqo/scalability/linear20_memory};
\nextgroupplot[ymode=log,ylabel={Network\\(bytes)}]
\addplot table[x index=0,y index=1] {plotsdata/moqo/scalability/linear16_network};
\addplot table[x index=0,y index=1] {plotsdata/moqo/scalability/linear18_network};
\addplot table[x index=0,y index=1] {plotsdata/moqo/scalability/linear20_network};
\end{groupplot}
\end{tikzpicture}

\ref{linearMoqoScalabilityLegend}
\caption{MPQ scales steadily using up to 256 workers for linear plan spaces and two plan cost metrics.\label{linearMoqoScalabilityFig}}
\end{figure}
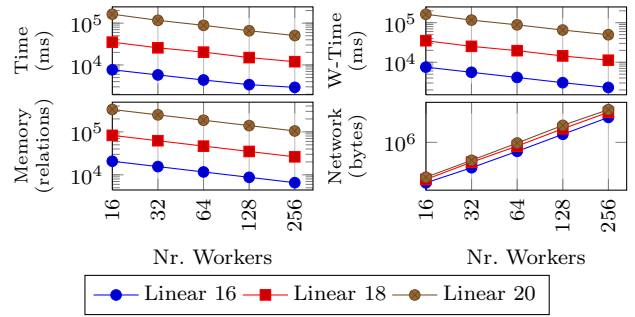

\section{Conclusion}
\label{conslusionSec}

We presented a generic plan space decomposition method for query optimization that is applicable for single- and multi-objective query optimization and for other variants. We demonstrated scalability using up to 256 workers. 



\begin{tiny}
\bibliographystyle{abbrv}

\end{tiny}

\end{document}